\documentclass[11pt,a4paper]{llncs}

\usepackage[english]{babel}
\usepackage[T1]{fontenc}
\usepackage{amssymb}

\newcommand\hit{\ensuremath{\cal H}}

\newcommand\zz{\ensuremath{\mathbb{Z}}}
\newcommand\nn{\ensuremath{\mathbb{N}}}
\newcommand\qq{\ensuremath{\mathbb{Q}}}

\newcommand{\minuszero}{\setminus \{0\}}
\begin{document}

\mainmatter
\title{A Hitting Set Construction, with Applications to Arithmetic Circuit Lower Bounds}

\author{Pascal Koiran}

\institute{ \today\\
LIP\thanks{UMR 5668 ENS Lyon, CNRS, UCBL, INRIA.}, \'Ecole Normale Sup\'erieure de Lyon, Universit\'e de Lyon\\
Department of Computer Science, University of Toronto\thanks{A part of this work was done during a visit to the Fields Institute.}
{\tt Pascal.Koiran@ens-lyon.fr} 
}

\titlerunning{Lower bounds from hitting sets}
\maketitle

\begin{abstract}
A polynomial identity testing algorithm must determine whether a given input polynomial is identically equal to 0.
We give a deterministic black-box 
identity testing algorithm for univariate polynomials of 
the form $\sum_{j=0}^t c_j X^{\alpha_j} (a + b X)^{\beta_j}$.
From our algorithm we derive an exponential lower bound for representations of
polynomials such as $\prod_{i=1}^{2^n} (X^i-1)$ under this form. 

It has been conjectured that these polynomials are hard to compute by general arithmetic circuits.
Our result shows that the ``hardness from derandomization'' approach to lower bounds 
is feasible for a restricted class of arithmetic circuits.
The proof is based on techniques from algebraic number theory, and more precisely on properties of the height function of algebraic numbers.
\end{abstract}

\section{Introduction}

The large body of work on hardness versus randomness tradeoffs shows that the 
two tasks of proving lower bounds and derandomizing algorithms are roughly 
equivalent. 
This equivalence holds both in the boolean and arithmetic world.
We focus here on the arithmetic world~\cite{KI04}.
The equivalence between lower bounds and derandomization 
suggests a new approach to lower bounds (see e.g.~\cite{KI04,Agra05}): 
let us derandomize algorithms first, and much-coveted lower bounds will follow.
This ``hardness from derandomization'' approach is very appealing, 
but apparently has not yet led to many new
lower bound results. There have been some recent advances in derandomization,
however, especially for identity testing of small-depth arithmetic circuits, e.g.~\cite{SS09,KS08} and for the more difficult problem of black-box circuit 
reconstruction~\cite{KS09}. Also techniques have been developed for obtaining
simultaneously lower bounds and identity tests~\cite{Saxena08}, thereby 
reinforcing the intuition that these two problems are intimately connected.

In this paper we use the ``hardness from derandomization'' approach to
obtain lower bounds for a certain class of arithmetic circuits.
More precisely, we prove lower bounds for representations of univariate 
polynomials under the form
\begin{equation} \label{expression}
\sum_{j=0}^t c_j X^{\alpha_j} (a + b X)^{\beta_j},
\end{equation}
where the $c_j$, $a$ and $b$ are rational numbers.
Polynomials of this form were first considered in~\cite{KaKoi05} due to their role in the factorization of sparse bivariate polynomials. Indeed, such an expression vanishes identically if and only if $Y-a-bX$ is a linear factor of the bivariate
polynomial $\sum_{j=0}^t c_j X^{\alpha_j} Y^{\beta_j}$.

Obviously, any univariate polynomial 
can be expressed under form~(\ref{expression}) by expanding it as a sum
of monomials (the resulting $\beta_j$ are all 0). 
Representation~(\ref{expression}) can potentially
be much more compact than the ``sum of monomials'' representation, however,
due to the presence of the possibly large exponents 
$\alpha_j$ and $\beta_j$ (note that $X^{\alpha_j}$ can be computed in 
about $\log \alpha_j$ multiplications by repeated squaring; the same
trick applies of course to $(a + b X)^{\beta_j}$).
The presence of possibly large exponents makes lower bounds and deterministic 
identity testing nontrivial.

\subsection{Lower Bound Statement} \label{statement}

A simple version of our lower bound result is as follows.
\begin{theorem} \label{simple_lb}
Consider a family of polynomials $(P_n)$ of the form
\begin{equation} \label{hardpoly}
P_n=\prod_{i=1}^{N_n} (X^i-1).
\end{equation}
Assume that $P_n$ can be expressed under form~(\ref{expression}) 
with $t$ polynomially bounded in $n$ 
and the 
bit sizes of the $c_j$, $\alpha_j$ and $\beta_j$ 
polynomially bounded in $n$.
Then $N_n$ must be polynomially bounded in $n$ as well.
\end{theorem}
We define the bit size of $c_j$ as the sum of the bit sizes of its
numerator and denominator.
Note that there is no restriction on the size of the coefficients $a$ and $b$ 
in this theorem (they may grow arbitrarily fast as a function of $n$).
Here we have expressed our result as a function of a single parameter~$n$
for the sake of clarity. We give in Theorem~\ref{main} a more precise (and slightly more general) lower 
bound where the dependency on each parameter is worked out carefully.
In particular, we work with the projective height $H(c)$ 
of the tuple $c=(c_j)$.
 This is a more appropriate notion of ``size'' of $c$ than the naive
bit size used in Theorem~\ref{simple_lb}.
The projective height is defined in Section~\ref{projective}.

The ``obvious'' arguments such as degree comparison 
between~(\ref{expression}) and~(\ref{hardpoly}) only show that $N_n$ must
be {\em exponentially} bounded in $n$.
Theorem~\ref{simple_lb} should therefore 
be viewed as an exponential lower bound.
One can also see the exponential nature of our lower bound by considering the polynomials 
$\prod_{i=1}^{2^n} (X^i-1)$: it follows from Theorem~\ref{main} that for some constant $\epsilon>0$, these polynomials cannot be expressed under form~(\ref{expression}) if~$t$ and the bit sizes of the $c_j$, $\alpha_j$ and $\beta_j$ are bounded by $2^{\epsilon n}$.

We note that the polynomials $P_n$ were suggested
by Agrawal as good candidates for proving lower bounds.
As observed by Agrawal~\cite{Agratalk}, 
if it could be shown that $P_n$ is hard to compute 
by general arithmetic circuits, it would follow that the permanent is hard to 
compute by arithmetic circuits. This also follows from a general result 
(Theorem~5 of~\cite{KoiPe07c}, see also~\cite{burgisser2007})
which roughly speaking shows the following: if the permanent has polynomial
size arithmetic circuits then exponential-size products of easy-to-compute 
polynomials are themselves easy to compute.

Note also that there is a formal similarity between~(\ref{hardpoly}) and the
well-known Pochhammer-Wilkinson polynomial $\prod_{i=1}^n (X-i)$ 
where roots of unity are replaced by integers.  The Pochhammer-Wilkinson
polynomial is widely conjectured to be hard to compute~\cite{burgisser2001,burgisser2007,Lipton94,ShubSmale}.
As explained in Section~\ref{remarks}, it is possible to obtain a good  lower bound for representations of this polynomial under form~(\ref{expression}).

\subsection{Main Ideas and Connections to Previous Work} \label{mainideas}

Our lower bound is based on the construction of hitting sets for polynomials
of the form~(\ref{expression}).
Recall that a hitting set $\hit$ for a set $\cal F$ of polynomials
is a (finite) set of points such that there exists for any non-identically 
zero polynomial $f \in \cal F$ at least one point $a \in \hit$ such that 
$f(a) \neq 0$. Hitting sets are sometimes called ``correct test sequences''~\cite{HS82}.
By a natural abuse of notation, 
we will sometimes say that $\hit$ is a hitting set for a polynomial $f$ 
if it a hitting set for the singleton $\{f\}$.

The existence of polynomial size hitting sets for general arithmetic circuits
follows from standard probabilistic arguments. 
A much more difficult problem is to give explicit (deterministic) constructions
of small hitting sets. It is easy to see that this problem is equivalent
to black-box deterministic identity testing: any hitting set for $\hit$
yields an obvious black-box identity testing algorithm 
(declare that $f \equiv 0$ iff $f$ evaluates to 0 on all the points of $\hit$);
conversely, assuming that $\cal F$ contains the identically zero  polynomial, the set of points queried by a black box algorithm on the input $f \equiv 0$
must be a hitting set for $\cal F$.

There is a general connection between lower bounds and derandomization
of polynomial identity testing~\cite{KI04}. This connection is especially
apparent in the case of black-box derandomization.
Namely, let $\hit$ be a hitting set for $\cal F$. The polynomial 
$P=\prod_{a \in \hit} (X-a)$ cannot belong to $\cal F$ since it is nonzero
and vanishes on $\hit$. The same remark applies to all nonzero multiples of $P$. If $\cal F$ is viewed as some
kind of ``complexity class'', we have therefore obtained a lower bound against $\cal F$ by exhibiting a polynomial $P$ which does not belong to $\cal F$.
This connection between hitting sets and arithmetic lower bounds has been known for at least 30 years~\cite{HS82}, but has led
to suprisingly few lower bound results.\footnote{As already observed in~\cite{HS82}, hitting sets may be difficult to construct precisely because they yield lower bounds.}
To the best of our knowledge, only one lower bound of this type is known: Agrawal~\cite[Corollary~65]{Agra05} has shown that certain multilinear polynomials cannot be computed by circuits with unbounded fanin addition gates of size $n^{2-\epsilon}$ and depth $(2-\epsilon)\log n$. The lower bound applies to polynomials with coefficients computable in PSPACE 
(this complexity class was independently defined in~\cite{KoiPe07a}, where it is called VPSPACE; further results on this class 
and other space-bounded classes in Valiant's model can be found in~\cite{KoiPe07b,MahaRao09,Poizat08}).

We have pointed out in Section~\ref{statement} that a lower bound for $P_n$
against general arithmetic circuits would imply 
a lower bound for the permanent. 
For the same reason (Theorem~5 of~\cite{KoiPe07c}), 
a hitting set construction against 
general arithmetic circuits would imply a lower bound for the permanent.

Our hitting set construction builds on work from~\cite{KaKoi05,KaKoi06}.
In~\cite{KaKoi05} we designed a deterministic identity testing algorithm
for expressions of the form~(\ref{expression}) as an intermediate 
step toward an algorithm for the factorization of ``supersparse'' bivariate
polynomials. Our identity testing algorithm was not black-box. Rather, 
it was based on a structure theorem (a so-called ``gap theorem'') 
which makes it possible to recognize easily identically zero expressions.
Here we build on this work to construct hitting sets. These sets turn out
to be made of roots of unity, explaining why we obtain a lower bound for 
polynomials of the form~(\ref{hardpoly}).

In terms of the class of arithmetic circuits studied, the work which seems 
closest to ours is by Saxena~\cite{Saxena08}. 
He gives lower bounds and identity testing algorithms 
for ``diagonal circuits'', 
i.e., sums of powers of (multivariate) 
linear functions, and more generally 
for sums of products of a small number of powers of linear functions.
Our circuits fall in this category since they compute sums of products of 
two powers of linear functions.
Our results and methods are quite different, however. 
He uses non-black-box methods, whereas we use black-box methods.
Moreover, his lower bounds break down for powers of high degree 
whereas we can handle high degree powers (indeed, for univariate polynomials
the only challenge is to prove lower bounds for polynomials of high degree
since any low degree polynomial can be represented efficiently as a sum of monomials, assuming that field constants are given for free).

\subsection{Organization of the paper}

As in~\cite{KaKoi05,KaKoi06} we use number-theoretic techniques and in particular properties of the height of algebraic numbers.
Some background on the height function
is provided in Section~\ref{background}.
Section~\ref{height_section} is technical: we obtain a height lower bound
which we use in Section~\ref{hitting} to construct our hitting sets.
From there, the lower bound theorem of Section~\ref{lb_section} follows
easily from the approach outlined in Section~\ref{mainideas}.
Finally, we suggest some possible extensions 
of our results in Section~\ref{remarks}.

\section{Number Theory Background} \label{background}

In this section we provide some background on the height function, 
first for algebraic numbers and then more generally 
for points in projective space.

\subsection{Heights of Algebraic Numbers} \label{height}

For any prime number~$p$, %
the $p$-adic absolute value on $\qq$ is characterized by
the following properties:
$|p|_p=1/p$, and $|q|_p=1$ if $q$ is a prime number different from
$p$.
For any $x \in \qq\minuszero$, %
$|x|_p$ can be computed as follows:
write $x=p^{\alpha}y$
where $p$ is relatively prime to the numerator and denominator of $y$,
and $\alpha \in \zz$. Then $|x|_p = 1/p^{\alpha}$
(and of course $|0|_p = 0$).
We denote by $M_{\qq}$ the union of the set of $p$-adic absolute values
and of the usual (archimedean) absolute value on $\qq$.

Let $d, e \in \zz$ %
be two non-zero relatively prime integers.
By definition, %
the height of the rational number $d/e$ is $\max(|d|$, $|e|)$. %
There is an equivalent definition in terms of absolute values:
for $x \in \qq$,
$H(x) = \prod_{\nu \in M_{\qq}} \max(1,|x|_{\nu})$.
Note in particular that $H(0)=1$.

More generally,
let $K$ be a number field (an extension of $\qq$ of finite degree).
The set $M_K$ of {\em normalized absolute values} is the set of
absolute values on $K$ which extend an absolute value of $M_{\qq}$.
For $\nu \in M_K$, we write $\nu | \infty$ if $\nu$ extends the usual
absolute value, and $\nu | p$ if $\nu$ extends the $p$-adic
absolute value.
One defines a ``relative
height'' $H_K$ on $K$ by the formula
\begin{equation} \label{heightdef}
H_K(x) = \prod_{\nu \in M_K} \max(1,|x|_{\nu})^{d_{\nu}}.
\end{equation}
Here $d_{\nu}$ is the so-called ``local degree''.
For every $p$
(either prime or infinite), $\sum_{\nu | p} d_{\nu} = [K:\qq]$.
The absolute height $H(x)$ of $x$ is $H_K(x)^{1/n}$, where $n=[K:\qq]$.
It is independent of the choice of $K$.
The above material is standard in algebraic number theory.
More details 
can be found for instance in \cite{Lang} or \cite{Wal00}.
We will also need a special case of a result due Amoroso and Zannier and already used in~\cite{KaKoi06}.
\begin{lemma} \label{AmoZannier}
Let $\theta$ be a root of unity and $a,b \in \qq$ such that $\alpha=a+b\theta$
is not a root of unity. If $\alpha \neq 0$ we have $H(\alpha) \geq C$ 
where $C>1$ is an absolute constant.
\end{lemma}
\begin{proof}
This follows from Theorem~1.1 of~\cite{AmZan00}
since the cyclotomic extension $\qq(\theta)$ is Abelian over $\qq$ 
(see for instance~\cite{vdWae40}, Section 8.4).$\Box$
\end{proof}

\subsection{Projective Height} \label{projective}

One can define a notion of (relative) height for
a point $c = (c_0,\ldots, c_t)$ in $K^{t+1}$ by the formula
$$H_K (c)=\prod_{\nu \in M_{K}} |c|_{\nu}^{d_{\nu}},$$
where $|c|_{\nu} = \max_{0 \leq j \leq t} |c_j|_{\nu}$.
This is the classical notion of height for a point in projective space
(\cite{HiSi00}, section~B.2).
 As a projective notion, $H_K(c)$ should be invariant by scalar multiplication.
Indeed, for $\lambda \in K\minuszero$
we have $H_K(\lambda c)=H_K(c)$.
This follows from the product formula:
$$\prod_{\nu \in M_K} |\lambda|_{\nu}^{d_{\nu}} = 1$$
for any $\nu \in K\minuszero$.
Note also that the (relative) height of an algebraic number $x \in K$ is equal to the projective height of the point $(1,x) \in K^2$.
As in the previous section, we can define an absolute height by the formula $H(c)=H_K(c)^{1/n}$ where $n=[K:\qq]$ and $K$ is 
chosen so that $c \in K^{t+1}$.

In our main lower bound theorem (Theorem~\ref{main}) we measure the size of the rational tuple $c=(c_j)$ in~(\ref{expression}) by its projective height instead of the naive bit size used in Theorem~\ref{simple_lb}. 
To compute the height of a rational tuple, we first note that 
$H(c)= \max_j |c_j|$ if the $c_j$ are relatively prime integers.
The general case $c_j \in \qq$ is therefore quite
easy: reduce to the same denominator to obtain integer coefficients,
divide by their gcd and take the maximum of the
absolute values of the resulting integers
(so in particular $H(c) \in \nn$ %
for any $c$ in~$\qq^{t+1}$).

\section{A Height Lower Bound} \label{height_section}

The goal of this section is to establish the following lower bound.
\begin{proposition} \label{hlb}
Let $(a,b)$ be a pair of rational numbers different from  the five ``excluded pairs'' $(0,0)$, $(\pm 1,0)$ 
and $(0,\pm 1)$.

There is a universal constant 
$C > 1$ such that the inequality 
\begin{equation} \label{hlb_equation}
H(a+b\theta) \geq C
\end{equation}
 holds for any root of unity $\theta$ which is not a 6th root of unity.
\end{proposition}
The inequality $H(a+b\theta) \geq C$ implies in particular that $a+b\theta$
is not a root of unity, since roots of unity are of height~1.

The main tool in the proof of Proposition~\ref{hlb} is the
 height lower bound of Lemma~\ref{AmoZannier}.
In light of this lemma, to complete the proof of Proposition~\ref{hlb} 
we just  need to understand when $a+b\theta$ can be a root of unity.
\begin{lemma}
Let $\theta$ be a root of unity and $(a,b)$ a pair of rational numbers different from  the five excluded pairs $(0,0)$, $(\pm 1,0)$ 
and $(0,\pm 1)$.
If $\theta$ is not a 6th root of unity then $\alpha = a+b\theta$ is nonzero, and is not a root of unity.
\end{lemma}
\begin{proof}
We will need some properties of cyclotomic polynomials. 
Recall that if $\theta$ is a root of unity of order $n$, its minimal polynomial
is the cyclotomic polynomials $\psi_n$. 
By definition, the conjugates of $\theta$ are the other roots of its minimal 
polynomial.
The roots of $\psi_n$ are exactly
the roots of unity of order $n$. 
There are $\phi(n)$ such roots, where $\phi(n)$ is Euler's totient function.
It is known that $\phi(n) \geq \sqrt{n}$ for $n>6$.
We therefore have $\phi(n) \geq 3$ for $n>6$.
From this it follows that $\phi(n) \geq 3$ except for $n=1,2,3,4$ or $6$.

The conclusion of the lemma clearly holds true in the case $b=0$. We therefore
assume in the remainder of the proof that $b \neq 0$.

The only rational roots of unity are $+1$ and $-1$, which are 6th roots of unity, hence $\alpha \neq 0$.
If  both $\theta$ and $a+b\theta$
happen to be  roots of unity then $\theta$ lies at the
intersection of the unit circle of the complex plane, and of the circle
defined by the condition $|a+bz|=1$. By excluding the 5 excluded pairs,
we have made sure that these two circles are distinct. They have therefore 
at most 2 intersection points. 
If $\theta'$ is a conjugate of $\theta$, the point $a+b\theta'$ is also 
a root of unity and must therefore lie at the intersection of the two circles.
Since there are at most two intersection points, $\theta$ has at most one conjugate.
This happens only when $\theta$ is a root of a cyclotomic polynomial
$\psi_n$ of degree $\phi(n) \leq 2$, and we have seen that there are 
only 5 possible values for $n$.
The two roots of order 4, $\pm i$, can be ruled out since $a\pm bi$ is a 
root of unity only when $(a,b)$ is equal to one of the two excluded pairs 
$(0,\pm 1)$.
We are left with the roots of unity of order 1, 2, 3 or 6, that is, with
the 6th roots of unity.
\qed \end{proof}

\begin{remark}
If $\theta^6=1$, $a+b\theta$ can be a root of unity for appropriate values 
of $a$ and $b$. For instance, if $\theta=e^{i\pi/3}$ then 
$1-\theta=e^{-i\pi/3}$.
If $\theta=e^{2i\pi/3}$ then $1+\theta = e^{i\pi/3}$.
\end{remark}

\begin{remark}
In the remainder of this paper we will apply~(\ref{hlb_equation}) only to 
$p$-th roots of unity where $p$ is prime.
\end{remark}

\section{Hitting Set Construction} \label{hitting}

It is well known that roots of unity yield hitting sets for sparse polynomials.
\begin{lemma} \label{sparse_divisors}
Let $K$ be a field of characteristic 0 and $f \in K[X]$ a nonzero univariate polynomial of degree at most $d$ with at most $m$ nonzero monomials. Then there are less
than $m \log d$ prime numbers $p$ for which $f(X)$ is identically zero 
modulo $X^p-1$.
\end{lemma}
Here we restrict to  fields of characteristic 0 but this lemma is 
stated in~\cite{BHLV09} for arbitrary integral domains.
A multivariate version can be found in Lemma~5 of~\cite{KaKoi06}.
Lemma~\ref{sparse_divisors} 
can be immediately restated in the language of hitting sets:
\begin{lemma} \label{sparse_hit}
Let $K$ be a field of characteristic 0, $\cal P$ a set of at least  $m \log d$ prime numbers and $\hit$ the set of all $p$-th roots of unity 
(in the algebraic closure of $K$) for all $p \in \cal P$. 

Then $\hit$ is a hitting set for the set of all polynomials $f \in K[X]$
of degree at most $d$ with at most $m$ nonzero monomials.
\end{lemma}
In the next proposition and theorems, the projective height comes into play.
Recall that this notion is defined in Section~\ref{background};
in particular, we explain at the end of that section how to compute 
$H(c)$ when the $c_j$ are rational 
(which is the case in Theorem~\ref{gap_th}).
For a rational tuple, the logarithm of the projective height gives 
a more appropriate notion of ``size'' than the naive bit size.
In the next proposition, we use the projective height for tuples of 
algebraic numbers. Namely, following Lenstra~\cite{Len99a} we define the height
$H(p)$ of a polynomial $p=\sum_{j=0}^t c_j X^j \in \overline{\qq}[X]$ 
as the projective height $H(c)$.
\begin{proposition} \label{lenstra}
Let $p \in \overline{\qq}[X]$ be a polynomial with at most $t+1$ non-zero terms.
Assume that $p$ can be written as the sum of two polynomials $q$ and $r$
where each monomial of $q$ has degree at most $\beta$ and each monomial of $r$
has degree at least $\gamma$.
Let $x \in \overline{\qq}^*$ be a root of $p$ that %
is not a root of unity.
If $\gamma - \beta > \log(t \, H(p))/\log H(x)$
then $x$ is a common root of $q$ and $r$.
\end{proposition}
The proof of Proposition~\ref{lenstra} can be found in~\cite{KaKoi06}.
It is essentially the same as the proof
of Proposition~2.3 of~\cite{Len99a}.

\begin{theorem}[Gap Theorem for Hitting Sets] \label{gap_th}
Let $f \in \qq[X]$ be a polynomial of the form~(\ref{expression}), with $(a,b)$ different from the five excluded pairs of Proposition~\ref{hlb}.
Assume without loss of generality that the sequence $(\beta_j)$ is nondecreasing, and assume also there exists $l$ such that
\begin{equation} \label{gap_bound}
\beta_{l+1} - \beta_l > \log(t(t+1)H(c)) / \log C
\end{equation}
where $C$ is the constant of Proposition~\ref{hlb},
and $H(c)$ is the projective height of the tuple $c=(c_j)$.

Let $\hit$ be a set of roots of unity with $\theta ^6 \neq 1$ 
for all $\theta \in \hit$.

Let $g=\sum_{j=0}^l c_j X^{\alpha_j} (a + b X)^{\beta_j}$
and $h=\sum_{j=l+1}^t c_j X^{\alpha_j} (a + b X)^{\beta_j}$.
If $\hit$ is a hitting set for $g$ and $h$, \hit\ is also a hitting set for 
$f=g+h$.
\end{theorem}
\begin{proof}
We need to show that $f(\theta)=0$ for all $\theta \in \hit$ implies $f=0$.
If $\theta \in \hit$ is a root of $f$ then $a+b\theta$ 
is a root of the univariate polynomial 
$p(X)=\sum_{j=0}^t c_j \theta^{\alpha_j} X^{\beta_j}$.
The height of $p$ satisfies the inequality $H(p) \leq (t+1)H(c)$.
The factor $t+1$ is due to the fact that each
monomial of $p$ ``comes'' from at most $t+1$ terms of~(\ref{expression});
see~\cite{KaKoi06}, Lemma 3 for a proof.
Since $\theta^6 \neq 1$ we have $H(a+b\theta) \geq C >1$ by Proposition~\ref{hlb}. We can therefore apply Proposition~\ref{lenstra}, and it follows that
$x=a+b\theta$ is a common root of the two univariate polynomials 
$q=\sum_{j=0}^l c_j \theta^{\alpha_j} X^{\beta_j}$ and
$r=\sum_{j=l+1}^t c_j \theta^{\alpha_j} X^{\beta_j}$.
This means exactly that $g(\theta)=h(\theta)=0$.

If these two equalities apply to every $\theta \in \hit$ we have $g=h=0$
since $\hit$ is supposed to be a hitting set for both $g$ and $h$.
Hence $f=g+h=0$.$\Box$
\end{proof}
We are now ready to state our main hitting set theorem.
The bound will depend on 3 parameters:
\begin{itemize}
\item[(i)] the parameter $t$ in~(\ref{expression}).
\item[(ii)] $d$, the maximal value of the $\alpha_j$.
\item[iii)] an upper bound $M$ on the projective height $H(c)$ 
of the tuple $c$.
\end{itemize}
Given $t$, $d$ and $M$ we define 
\begin{equation} \label{delta}
\delta=\log(t(t+1)M) / \log C.
\end{equation}
Notice that this is essentially the gap bound in~(\ref{gap_bound}).
\begin{theorem}[Hitting Set Construction] \label{hit_theorem}
Let $\cal P$ be a set of at least $(t+1)(\delta t +1) \log(d+t\delta)$ 
prime numbers,
with $\delta$ as in~(\ref{delta}) and $p \geq 5$ for all $p \in \cal P$.

Let $\hit$ be the set of all $p$-th roots of unity for all $p \in \cal P$.
Then $\hit$ is a hitting set for the set of polynomials that can be represented
under form~(\ref{expression}) with $\alpha_j \leq d$ for all $j$, 
the rational tuple $c$ of projective height $H(c) \leq M$, 
and $(a,b)$ different 
from the two pairs $(0,\pm 1)$.
\end{theorem}
\begin{proof}
We proceed by reduction to Lemma~\ref{sparse_hit}.
As in Theorem~\ref{gap_th}, we will assume without loss of generality that the sequence $(\beta_j)$ is nondecreasing.
We can of course assume that $(a,b) \neq (0,0)$ 
since the corresponding polynomial in~(\ref{expression}) would be identically 
zero. We will also assume that $(a,b) \neq (\pm 1,0)$. In that case, 
$f$ can be written as a sum of $t+1$ monomials of degree at most $d$ 
and we can apply Lemma~\ref{sparse_hit}: $\hit$ is a hitting set for $f$
since $|{\cal P}| \geq (t+1)\log d$ 
(the same argument could of course 
be applied to any pair $(a,b)$ with $b=0$).

The remainder of the proof is divided in two cases. We first consider the case where there is no gap in $f$ in the sense of Theorem~\ref{gap_th}, that is,
$\beta_{l+1} - \beta_l \leq \delta$ for all $l$. 
In this case, factoring out the polynomial $(a+bX)^{\beta_0}$ if necessary, 
we assume without  loss of generality that $\beta_0=0$.
This is legitimate since the nonzero polynomial $(a+bX)^{\beta_0}$ does
not vanish at any point of $\hit$ (recall that the elements of $\hit$ are irrational numbers). From the relations $\beta_0=0$ 
and $\beta_{l+1} - \beta_l \leq \delta$ we find that 
$\beta_t = \max _l \beta_l \leq \delta t$.
Expanding each factor $(a+bX)^{\beta_j}$ in~(\ref{expression}) as a sum of monomials, we see that $f$ can be written as a sum of at most 
$(t+1)(\delta t +1)$ monomials,
each of degree at most $d+t\delta$. Lemma~\ref{sparse_hit} 
therefore implies that $\hit$ is a hitting set for $f$.

We finally consider the case where there are gaps in $f$. By ``breaking $f$ at the gaps'', we write $f=\sum_{i=1}^s f_i$ where each $f_i$ 
is a sum of consecutive terms $c_j X^{\alpha_j}(a+bX)^{\beta_j}$ 
from~(\ref{expression}). More precisely,
 we make sure that there is no gap inside
each $f_i$ in the sense that the difference between two consecutive exponents
$\beta_j$ in $f_i$ is bounded by $\delta$, and there is a gap between $f_i$ and $f_{i+1}$ in the sense that the difference between the smallest exponent 
$\beta_j$ in $f_{i+1}$ and the biggest one in $f_i$ is greater than $\delta$.

We have seen that $\hit$ is a hitting set for each of the $f_i$.
Applying Theorem~\ref{sparse_hit} repeatedly ($s-1$ times), we see that $\hit$ is a hitting
set for $f$ as well.$\Box$
\end{proof}

\begin{remark}
The pair $(a,b)=(0,\pm 1)$ is excluded from Theorem~\ref{hit_theorem}.
This case can easily be handled with Lemma~\ref{sparse_hit}: 
$f$ is a sum of $t+1$ monomials of degree at most $d+d'$, where $d'=\max_j \beta_j$. We can therefore replace the set $\cal P$ in Theorem~\ref{hit_theorem} by a set of prime numbers of cardinality at least $(t+1) \log (d+d')$.
By contrast, the bound in Theorem~\ref{hit_theorem} does not depend on~$d'$.
Also, we can construct a single hitting set which covers uniformly the
two cases $(a,b) \neq (0,\pm 1)$ and $(a,b) = (0,\pm 1)$ by replacing the bound $(t+1)(\delta t +1) \log(d+t\delta)$ in Theorem~\ref{hit_theorem} by the maximum of this bound and $(t+1) \log (d+d')$.
\end{remark}

\section{Lower Bound Theorem} \label{lb_section}

As explained in Section~\ref{mainideas}, it is straightforward 
to obtain a lower from our hitting set construction.
\begin{theorem}[Main Lower Bound] \label{main}
Let $\cal P$ be a set of prime numbers with $p \geq 5$ for all $p \in \cal P$,
\begin{equation} 
|{\cal P}| \geq (t+1)\max(\log (d+d'),(\delta t +1) \log(d+t\delta))
\end{equation}
and $\delta$ as in~(\ref{delta}).
The polynomial
$$P=\prod_{i \in {\cal P}} (X^{p_i}-1)$$
cannot be expressed under form~(\ref{expression}) if $\alpha_j \leq d$ and
$\beta_j \leq d'$ for all $j$, 
and if the rational tuple $c$ is of projective height $H(c) \leq M$.
The same lower bound applies to all nonzero multiples of $P$.
\end{theorem}
\begin{proof}
Let $f$ be a polynomial which can be expressed under form~(\ref{expression})
with $\alpha_j \leq d$ and $\beta_j \leq d'$ for all $j$, and $H(c) \leq M$.
Let $Q$ be a multiple of~$P$. By Theorem~\ref{hit_theorem} 
and the remark following it, the set of roots of $Q$ is a hitting set for $f$.
Hence we cannot have $f=Q$, unless $Q=0$.$\Box$
\end{proof}
Theorem~\ref{simple_lb} follows from Theorem~\ref{main} since
there are $\Omega(N/\log N)$ prime numbers in the interval $[2,N]$.

\section{Further Remarks} \label{remarks}

One can try to extend our results in various ways.
One possible direction is prove lower bounds for other polynomials than polynomials of the form $\prod_i (X^i-1)$.
It was recently shown in~\cite{Aven09} that nonzero polynomials represented under form~(\ref{expression}) have at most $6t-4$ real roots.
As a result, any set of $6t-3$ real numbers is a hitting set and we have lower bounds for polynomials with many real roots such as $\prod_{i=1}^{2^n}(X-i)$.

Perhaps more importantly, one can look for lower bounds under more general representations than~(\ref{expression}).
We make two suggestions below.
\begin{enumerate}

\item Consider expressions of the form 
\begin{equation} \label{change}
\sum_{j=0}^t c_j (a+bX)^{\alpha_j} (c+dX)^{\beta_j}.
\end{equation}
Assuming that $b \neq 0$, 
the change of variable $Y=a+bX$
brings us back to~(\ref{expression}) and we can use the black-box  algorithm of the present paper or the non-black-box algorithm of~\cite{KaKoi05} to perform deterministic identity testing.
Unfortunately, the change of variable $Y=a+bX$ is non-black-box and as a result
we do not have a lower bound for polynomials of the form $\prod_i (X^i-1)$.
Nevertheless, the set of real numbers is invariant under this change of variable. As a result, it follows again from~\cite{Aven09} that any set of $6t-3$ real numbers is a hitting set for~(\ref{change}) and we still have exponential lower bounds 
for polynomials such as $\prod_{i=1}^{2^n}(X-i)$.

The case $b=0$ is even simpler: now we  have polynomials
of the form 
$$\sum_{j=0}^t c'_j (c+dX)^{\beta_j},$$
where $c'_j = c_j a^{\alpha_j}$.
The change of variable $Y=c+dX$ shows that by Descarte's rule of signs, 
such a polynomial can have at most $2t+1$ real roots if it is nonzero. We can therefore construct a hitting set (any set of $2t+2$ real numbers will do)
and derive good lower bounds.

\item Consider now expressions of the form 
$$\sum_{j=0}^t c_j X^{\alpha_j} (a_j + b_j X)^{\beta_j}.$$
In~(\ref{expression}) we have $a_j = a$ and $b_j = b$ for all $j$.
Is deterministic identity testing feasible, either in a black-box 
or non-black-box way~? Is it possible to derive lower bounds for this form 
of polynomial representation~?
\end{enumerate}

\small

\section*{Acknowledgments} This work was to a large extent triggered by a question of Erich Kaltofen: can the polynomial $(X^n-1)/(X-1)$ be represented efficiently under form~(\ref{expression}) ?

\bibliographystyle{plain}

\begin{thebibliography}{10}

\bibitem{Agratalk}
M.~Agrawal.
\newblock A possible pseudorandom generator against arithmetic circuits.
\newblock Talk at the Daimi Workshop on Algebraic Complexity Theory. Aarhus,
  September 2008.

\bibitem{Agra05}
M.~Agrawal.
\newblock Proving lower bounds via pseudo-random generators.
\newblock In {\em Proc. FSTTCS 2005}.
\newblock Invited survey.

\bibitem{AmZan00}
F.~Amoroso and U.~Zannier.
\newblock A relative {Dobrowolski} lower bound over {Abelian} varieties.
\newblock {\em Ann. Scuola Norm. Sup. Pisa Cl. Sci. 4}, 29(3):711--727, 2000.


\bibitem{Aven09}
Avendano, M.
\newblock The number of roots of a lacunary bivariate polynomial on a line.
\newblock {\em Journal of Symbolic Computation}, 44:1280--1284, 2009.

\bibitem{BHLV09}
M.~Bl\"aser, M.~Hardt, R.~J. Lipton, and N.~K. Vishnoi.
\newblock Deterministically testing sparse polynomial identities of unbounded
  degree.
\newblock {\em Information Processing Letters}, 109(3):187--192, 2009.

\bibitem{burgisser2001}
P.~B\"urgisser.
\newblock On implications between {P-NP} hypotheses: decision versus
  computation in algebraic complexity.
\newblock In {\em Proc. 26th International Symposium on Mathematical
  Foundations of Computer Science (MFCS 2001)}, pages 3--17. Springer, 2001.
\newblock Invited paper.

\bibitem{burgisser2007}
P.~B\"urgisser.
\newblock On defining integers in the counting hierarchy and proving lower
  bounds in algebraic complexity.
\newblock In {\em Proc. STACS 2007}, pages 133--144, 2007.
\newblock Full version: ECCC Report No.~113, August 2006.

\bibitem{HS82}
J.~Heintz and C.-P. Schnorr.
\newblock Testing polynomials which are easy to compute.
\newblock In {\em Logic and Algorithmic (an International Symposium held in
  honour of {Ernst Specker})}, pages 237--254. Monographie $n^{\tiny o}$ 30 de
  L'Enseignement Math\'ematique, 1982.
\newblock Preliminary version in {\em Proc. 12th {ACM} Symposium on Theory of
  Computing}, pages~262-272, 1980.

\bibitem{HiSi00}
M.~Hindry and J.~H. Silverman.
\newblock {\em Diophantine Geometry: an Introduction}, volume 201 of {\em
  Graduate Texts in Mathematics}.
\newblock Springer, 2000.

\bibitem{KI04}
V.~Kabanets and R.~Impagliazzo.
\newblock Derandomizing polynomial identity test means proving circuit lower
  bounds.
\newblock {\em Computational Complexity}, 13(1-2):1--46, 2004.

\bibitem{KaKoi05}
E.~Kaltofen and P.~Koiran.
\newblock On the complexity of factoring bivariate supersparse (lacunary)
  polynomials.
\newblock In {\em Proc. 2005 International Symposium on Symbolic and Algebraic
  Computation (ISSAC)}. {ACM Press}, 2005.

\bibitem{KaKoi06}
E.~Kaltofen and P.~Koiran.
\newblock Finding small degree factors of multivariate supersparse (lacunary)
  polynomials over algebraic number fields.
\newblock In {\em Proc. 2006 International Symposium on Symbolic and Algebraic
  Computation (ISSAC)}. {ACM Press}, 2006.

\bibitem{KoiPe07a}
P.~Koiran and S.~Perifel.
\newblock {VPSACE} and a transfer theorem over the reals.
\newblock In {\em Proc. STACS 2007}, volume 4393 of {\em Lecture Notes in
  Computer Science}, pages 417--428. Springer-Verlag, 2007.
\newblock Journal version to appear in {\em Computational Complexity}.

\bibitem{KoiPe07b}
P.~Koiran and S.~Perifel.
\newblock {VPSPACE} and a transfer theorem over the complex field.
\newblock In {\em Proc.32nd International Symposium on Mathematical Foundations
  of Computer Science}, volume 4708 of {\em Lecture Notes in Computer Science},
  pages 359--370. Springer, 2007.
  
  \bibitem{KS08}
Z.~S. Karnin and A.~Shpilka.
\newblock Black box polynomial identity testing of depth-3 arithmetic circuits
  with bounded top fan-in.
\newblock In {\em Proc. 23rd IEEE Conference on Computational Complexity
  (CCC)}, 2008.

\bibitem{KS09}
Z.~S. Karnin and A.~Shpilka.
\newblock Reconstruction of generalized depth-3 arithmetic circuits with
  bounded top fan-in.
\newblock In {\em Proc. 24th IEEE Conference on Computational Complexity
  (CCC)}, 2009.

\bibitem{KoiPe07c}
P.~Koiran and S.~Perifel.
\newblock Interpolation in {Valiant}'s theory, 2007.
\newblock http://arxiv.org/abs/0710.0360.

\bibitem{Lang}
S.~Lang.
\newblock {\em Algebra}.
\newblock Addison-Wesley, 1993.

\bibitem{Len99a}
H.~W. Lenstra.
\newblock Finding small degree factors of lacunary polynomials.
\newblock In {\em Number Theory in Progress}, pages 267--276, 1999.

\bibitem{Lipton94}
R.~J. Lipton.
\newblock Straight-line complexity and integer factorization.
\newblock In {\em Proc. First International Symposium on Algorithmic Number
  Theory}, volume 877 of {\em Lecture Notes in Computer Science}, pages 71--79.
  Springer, 1994.


\bibitem{MahaRao09}
M.~Mahajan and B.~V.~R. Rao.
\newblock Small-space analogues of {Valiant's} classes.
\newblock In {\em Proc. 17th International Symposium on Fundamentals of
  Computation Theory}, volume 5699 of {\em Lecture Notes in Computer Science},
  pages 250--261. Springer, 2009.
  
  \bibitem{Poizat08}
B.~Poizat.
\newblock \`{A} la recherche de la d\'efinition de la complexit\'e d'espace
  pour le calcul des polyn\^omes \`{a} la mani\`{e}re de {Valiant}.
\newblock {\em Journal of Symbolic Logic}, 73(4):1179--1201, 2008.

\bibitem{Saxena08}
N.~Saxena.
\newblock Diagonal circuit identity testing and lower bounds.
\newblock In {\em Proc. 35th International Colloquium on Automata, Languages
  and Programming (ICALP 2008)}, LNCS 5125, pages 60--71. Springer, 2008.

\bibitem{SS09}
N.~Saxena and C.~Seshadri.
\newblock An almost optimal rank bound for depth-3 identities.
\newblock In {\em Proc. 24th IEEE Confenrence on Computational Complexity
  (CCC)}, 2009.

\bibitem{ShubSmale}
M.~Shub and S.~Smale.
\newblock On the intractability of {Hilbert}'s {Nullstellensatz} and an
  algebraic version of ``{P=NP}".
\newblock {\em Duke Mathematical Journal}, 81(1):47--54, 1995.

\bibitem{vdWae40}
B.~L. van~der Waerden.
\newblock {\em Moderne Algebra}.
\newblock Springer Verlag, Berlin, 1940.
\newblock English transl. publ. under the title ``Modern algebra'' by F. Ungar
  Publ. Co., New York, 1953.

\bibitem{Wal00}
M.~Waldschmidt.
\newblock {\em Diophantine approximation on linear algebraic groups}.
\newblock Springer, 2000.

\end{thebibliography}

\end{document}